\documentclass[11pt]{amsart}
\usepackage{latexsym}
\usepackage{amsmath}
\usepackage{amssymb}
\usepackage{amsthm}
\usepackage{epsfig}
\usepackage{algorithm,algorithmic}
\usepackage[mathscr]{euscript}
\usepackage{tikz}
\usepackage{hyperref}
\usepackage{float}
\usepackage[sort]{cite}
\usepackage{caption}
\usepackage{fullpage}

\newcommand{\prob}[1]{\mathbb P\left [ #1 \right ]}

\DeclareMathOperator{\poly}{poly}

\DeclareMathOperator{\match}{\mathsf{match}}

\newtheorem{theorem}{Theorem}
\newtheorem{corollary}[theorem]{Corollary}

\begin{document}
\title{A Note on Logarithmic Space Stream Algorithms for Matchings in Low Arboricity Graphs}
\author{Andrew McGregor and Sofya Vorotnikova}
\thanks{Supported by NSF CAREER Award CCF-0953754 and CCF-1320719 and a Google Faculty Research Award.}
\address{College of Computer and Information Sciences, University of Massachusetts.}
\date{}
\maketitle

\section{Introduction}

We present a data stream algorithm for estimating the size of the maximum matching of a low arboricity graph. Recall that a graph has arboricity $\alpha$ if its edges can be partitioned into at most $\alpha$ forests and that a planar graph has arboricity $\alpha=3$. Estimating the size of the maximum matching in such graphs has been a focus of recent data stream research \cite{one,two,EsfandiariHLMO15,BuryS15a,ChitnisCEHMMV16}. See also \cite{McGregor14} for a survey of the general area of graph algorithms in the stream model.

A surprising result on this problem was recently proved by Cormode et al.~\cite{one}. They designed an ingenious algorithm that returned a $(22.5\alpha+6)(1+\epsilon)$ approximation using a single pass over the edges of the graph (ordered arbitrarily) and $O(\epsilon^{-2}\alpha \cdot \log n \cdot \log_{1+\epsilon} n)$ space\footnote{Here, and throughout, space is specified in words and we assume that an edge or a counter (between $0$ and $\alpha$) can be stored in one word of space.}. We improve the approximation factor to $(\alpha+2)(1+\epsilon)$ via a tighter analysis and show that, with a  modification of their algorithm, the space required can be reduced to  $O(\epsilon^{-2} \log n)$.

\section{Results}

Let $\match(G)$ be the maximum size of a matching in a graph $G$ and let $E_{\alpha}$ be the set of edges $uv$ where the number of edges incident to $u$ or $v$ that appear in the stream after $uv$ are both at most $\alpha$.

\subsection{A Better Approximation Factor} We first show a bound for $\match(G)$ in terms of  $|E_{\alpha}|$. Cormode et al.~proved a similar but looser bound.

\begin{theorem}\label{thm:bound}
$\match(G)\leq   |E_{\alpha}|\leq (\alpha+2) \match(G)$.
\end{theorem}

\begin{proof} We first prove the left inequality. To do this define $y_{e}=1/(\alpha+1)$ if $e$ is in $E_{\alpha}$ and 0 otherwise. Note that $y_{e}$ is a fractional matching with maximum weight $1/(\alpha+1)$ and hence\footnote{It can be shown as a corollary of Edmonds Matching Polytope Theorem \cite{Edm65} that any fractional matching in which all edge weights are bounded by $\epsilon$ is at most a factor $1+\epsilon$ larger than the maximum integral matching. See~\cite[Theorem 5]{two} for details.}\,
$$\frac{|E_{\alpha}|}{\alpha+1}= \sum_{e} y_{e} \leq \left(1 + \frac{1}{\alpha+1}\right)\match(G) = \frac{\alpha+2}{\alpha+1}\match(G) \ .$$ 

It remains to prove the right inequality. Define $H$ to be the set of vertices with degree $\alpha+1$ or greater. We refer to these as the \emph{heavy} vertices. For $u \in H$, let $B_u$ be the set of the last $\alpha+1$ edges incident to $u$ that arrive in the stream. 

Say an edge $uv$ is \emph{good} if $uv\in B_u \cap B_v$ and \emph{wasted} if $uv\in B_u \oplus B_v$, i.e., the symmetric difference. Then $|E_\alpha|$ is exactly the number of good edges. Define 
\begin{align*}
w &= \mbox{number of good edges with exactly no end points in $H$} ~,\\
x &= \mbox{number of good edges with exactly one end point in $H$} ~,\\
y &= \mbox{number of good edges with two end points in $H$} ~,\\ 
z &= \mbox{number of wasted edges with two end points in $H$} ~,
\end{align*}
and note that $|E_\alpha| = w+x+y$. 
 
We know $
x+2y+z=(\alpha+1)|H|
$ because $B_u$ contains exactly $\alpha+1$ edges if $u\in H$. Furthermore, 
 $
z+y \leq \alpha |H|$ because the graph has arboricity $\alpha$. 
Therefore \[x+y\geq (\alpha+1)|H|- \alpha |H|=|H| \ .\]
 Let $E_L$ be the set of edges with no endpoints in $H$. Since every edge in $E_L$ is good,
 $w=|E_L|$. 
 Hence, $|E_\alpha|\geq |H|+|E_L|\geq \match(G)$ where the last inequality follows because at most one edge incident to each heavy vertex can appear in a matching.
\end{proof}

Let $G_t$ be the graph defined by the stream prefix of length $t$ and let $E_{\alpha}^t$ be the set of good edges with respect to this prefix, i.e., all edges $uv$ from $G_t$ where the number of edges incident to $u$ or $v$ that appear  after $uv$ in the prefix are both at most $\alpha$. By applying the theorem to $G_t$, and noting that $E^*\geq |E_\alpha|$ and $\match(G_t)\leq \match(G)$, we deduce the following corollary:

\begin{corollary}\label{cor:bound2}
Let $E^*=\max_t |E_{\alpha}^t|$. Then $\match(G)\leq E^* \leq (\alpha+2) \match(G)$.
\end{corollary}

\subsection{A (Slightly) Better Algorithm.}
See Figure \ref{fig:firstalg} for an algorithm that approximates $E^*$ to a $(1+\epsilon)$-factor in the insert-only graph stream model. The algorithm is a modification of the algorithm for estimating $|E_{\alpha}|$ designed by Cormode et al.~\cite{one}. The basic idea is to independently sample edges from $E_\alpha^t$ with probability that is high enough to obtain an accurate approximation of $|E_{\alpha}^t|$ and yet low enough to use a small amount of space. For every sampled edge $e = uv$, the algorithm stores the edge itself and two counters $c^u_e$ and $c^v_e$ for degrees of its endpoints in the rest of the stream. If we detect that a sampled edge is not in $E_\alpha^t$, i.e., either of the associated counters exceed $\alpha$, it is deleted. 

Cormode et al.~ran multiple instances of this basic algorithm corresponding to sampling probabilities $1, (1+\epsilon)^{-1}, (1+\epsilon)^{-2}, \ldots $ in parallel; terminated any instance that used too much space; and returned an estimate based on one of the remaining instantiations. Instead, we start sampling with probability 1 and put a cap  on the number of edges stored by the algorithm. Whenever the capacity is reached, the algorithm halves the sampling probability and deletes every edge currently stored with probability $1/2$. This modification saves a factor of $O(\epsilon^{-1} \log n)$ in the space use and update time of the algorithm. We save a further $O(\alpha)$ factor in the analysis by using the algorithm to estimate $E^*$ rather than $|E_\alpha|$.
The proof of correctness is  similar to that for the original algorithm.

\begin{figure*}[t]
\begin{center}
\fbox{
\begin{minipage}{6.2in}
{
\noindent {\bf Algorithm 1: {\sc Approximating $E^*$}}
\begin{enumerate}\parskip=0in
\item Initialize $S\leftarrow \emptyset$, $p=1$, $\max=0$
\item For each edge $e=uv$ in the stream:
\begin{enumerate}
\item With probability $p$ add $e$ to $S$ and initialize counters $c^u_e \leftarrow 0$ and $c^v_e \leftarrow 0$
\item For each edge $e' \in S$, if $e'$ shares endpoint $w$ with $e$: 
\begin{itemize}
\item Increment $c^w_{e'}$
\item If $c^w_{e'} > \alpha$, remove $e'$ from $S$ and corresponding counters
\end{itemize}
\item If $|S| > 30 \epsilon^{-2} \log n$:
\begin{itemize}
\item $p \leftarrow p/2$
\item Remove each edge in $S$ and corresponding counters with probability $1/2$
\end{itemize}
\item $\max \leftarrow \max(\max,|S|/p)$
\end{enumerate}
\item Return $\max$
\end{enumerate}
}
\end{minipage}
}
\end{center}
\caption{{\sc Approximating $E^*$} Algorithm.}\label{fig:firstalg}
\end{figure*}

\begin{theorem}
With high probability, Algorithm \ref{fig:firstalg} outputs a $(1+\epsilon)$ approximation of $E^*$.
\end{theorem}

\begin{proof} 
Let $k$ be such that $2^{k-1}\tau \leq E^*<2^k \tau$ where $\tau = 20 \epsilon^{-2} \log n$. 
First suppose we toss $O(\log n)$ coins for each edge in $E_\alpha^t $ and say that an edge $e$ is sampled at level $i$ if at least the first $i-1$ coin tosses at heads. Hence, the probability that an edge is sampled at level $i$ is $p_i = 1/2^{i}$ and that the probability an edge is sampled at level $i$ conditioned on being sampled at level $i-1$ is $1/2$. Let $s_i^t$ be the number of edges sampled.
It follows from  the Chernoff bound that for $i\leq k$,
\begin{align*}
\prob{|s_i^t - p_i |E_\alpha^t| | \geq \epsilon p_i E^*} \leq \exp \left(- \frac{\epsilon^2  E^*p_i }{4} \right)
\leq  \exp\left( - \frac{\epsilon^2 E^*p_k}{4} \right)\leq \exp\left( - \frac{\epsilon^2 \tau}{8} \right)=1/\poly(n) \ .
\end{align*}
By the union bound, with high probability,  $s^t_i/p_i = |E_{\alpha}^t|\pm \epsilon E^*$ for all $0\leq i \leq k$, $1\leq t\leq n$. 


The algorithm initially maintains the edges in $E_\alpha^t$ sampled at level $i=0$. If the number of these edges exceeds the threshold, we subsample these to construct the set of edges sampled at level $i=1$. If this set of edges also exceeds the threshold, we again subsample these to construct the set of edges at level $i=2$ and so on.
If $i$ never exceeds $k$, then the above calculation implies that the output is $(1\pm \epsilon) E^*$. But if $s^{t}_k$ is bounded above by $(1+\epsilon)E^*/2^k < (1+\epsilon)\tau$ for all $t$ with high probability, then $i$ never exceeds $k$.
\end{proof}

It is immediate that the algorithm uses $O(\epsilon^{-2}  \log n)$ space since this is the maximum number of edges stored at any one time. By Corollary \ref{cor:bound2}, $E^*$ is an $(\alpha+2)$ approximation of $\match(G)$ and hence we have proved the following theorem.

\begin{theorem}
The size of the maximum matching of a graph with arboricity $\alpha$ can be $(\alpha+2)(1+\epsilon)$-approximated  with high probability using a single pass over the edges of $G$ given  $O(\epsilon^{-2} \log n)$ space.
\end{theorem}

\subsection*{Acknowledgement} In an earlier version of the proof of Theorem 3, we erroneously claimed that, conditioned on the current sampling rate being $1/2^j$, all edges in $E_\alpha^t$ had been sampled at that rate. Thanks to  Sepehr Assadi, Vladimir Braverman, Michael Dinitz, Lin Yang, and Zeyu Zhang for raising this issue.


\begin{thebibliography}{9}

\bibitem{BuryS15a}
Marc Bury and Chris Schwiegelshohn.
\newblock Sublinear estimation of weighted matchings in dynamic data streams.
\newblock In {\em Algorithms - {ESA} 2015 - 23rd Annual European Symposium,
  Patras, Greece, September 14-16, 2015, Proceedings}, pages 263--274, 2015.
  
  \bibitem{ChitnisCEHMMV16}
Rajesh Chitnis, Graham Cormode, Hossein Esfandiari, MohammadTaghi Hajiaghayi,
  Andrew McGregor, Morteza Monemizadeh, and Sofya Vorotnikova.
\newblock Kernelization via sampling with applications to finding matchings and
  related problems in dynamic graph streams.
\newblock In {\em Proceedings of the Twenty-Seventh Annual {ACM-SIAM} Symposium
  on Discrete Algorithms, {SODA} 2016, Arlington, VA, USA, January 10-12,
  2016}, pages 1326--1344, 2016.
  
  \bibitem{one} 
Graham Cormode, Hossein Jowhari, Morteza Monemizadeh, S. Muthukrishnan. 
\textit{The Sparse Awakens: Streaming Algorithms for Matching Size Estimation in Sparse Graphs}. 
CoRR, abs/1608.03118, 2016.
 
 \bibitem{Edm65}
Jack Edmonds.
\newblock Maximum matching and a polyhedron with 0,1-vertices.
\newblock {\em Journal of Research of the National Bureau of Standards,
  69:125-130}, 1965.
  
\bibitem{EsfandiariHLMO15}
Hossein Esfandiari, Mohammad~Taghi Hajiaghayi, Vahid Liaghat, Morteza
  Monemizadeh, and Krzysztof Onak.
\newblock Streaming algorithms for estimating the matching size in planar
  graphs and beyond.
\newblock In {\em Proceedings of the Twenty-Sixth Annual {ACM-SIAM} Symposium
  on Discrete Algorithms, {SODA} 2015, San Diego, CA, USA, January 4-6, 2015},
  pages 1217--1233, 2015.
  
  \bibitem{McGregor14}
Andrew McGregor.
\newblock Graph stream algorithms: a survey.
\newblock {\em {SIGMOD} Record}, 43(1):9--20, 2014.

\bibitem{two} 
Andrew McGregor and Sofya Vorotnikova. 
\textit{Planar Matching in Streams Revisited}. 
APPROX, 2016.

\end{thebibliography}
\end{document}